\documentclass{jocg}
\usepackage{amsmath,amsfonts}

\usepackage{amsthm}
\theoremstyle{plain}
\newtheorem{theorem}{Theorem}
\newtheorem{proposition}{Proposition}
\newtheorem{lemma}{Lemma}
\newtheorem{conjecture}{Conjecture}

\usepackage{graphicx}
\usepackage{xspace}

\usepackage[mathlines]{lineno}

\makeatletter
\let\@fnsymbol\@arabic
\makeatother

\title{%
  \MakeUppercase{Weight Balancing on Boundaries}%
}

\author {
  Luis~Barba,%
  \thanks{\affil{Institute of Theoretical Computer Science, ETH Z\"urich, Z\"urich, Switzerland},
          \email{luis.barba@inf.ethz.ch}}\,
  Otfried~Cheong,%
  \thanks{\affil{School of Computing, KAIST, Daejeon, Korea},
          \email{otfried@kaist.airpost.net}}\,
  Michael~Gene~Dobbins,%
  \thanks{\affil{Department of Mathematics, Binghamton University, Binghamton, NY, USA},
          \email{michaelgenedobbins@gmail.com}}\,
  Rudolf~Fleischer,%
  \thanks{\affil{Department of Computer Science, Heinrich Heine University D\"usseldorf, Germany},
    \email{rudolf.fleischer@gmail.com}}\,
  Akitoshi~Kawamura,%
  \thanks{\affil{Research Institute for Mathematical Sciences, Kyoto University, Kyoto, Japan},
          \email{kawamura@kurims.kyoto-u.ac.jp}}\,
  Matias~Korman,%
  \thanks{\affil{Siemens Electronic Design Automation, USA},
          \email{matias\_korman@mentor.com}}\,
  Yoshio~Okamoto,%
  \thanks{\affil{Department of Computer and Network Engineering, University of Electro-Communications, Tokyo, Japan},
          \email{okamotoy@uec.ac.jp}}\,
  J\'{a}nos~Pach,%
  \thanks{\affil{R\'enyi Institute, Budapest, Hungary and Moscow Institute of Physics and Technology, Moscow, Russia},
          \email{janos.pach@epfl.ch}}\,
  Yuan~Tang,%
  \thanks{\affil{Software School, Shanghai Key Laboratory of Intelligent Information Processing, Fudan University, Shanghai, P. R. China},
          \email{yuantang@fudan.edu.cn}}\,
  Takeshi~Tokuyama,%
  \thanks{\affil{School of Engineering, Kwansei Gakuen University,  Japan},
    \email{tokuyama@kwansei.ac.jp}}\quad 
  and Sander~Verdonschot\,%
  \thanks{\affil{Shopify, Ottawa, Ontario, Canada. The research for this paper was
performed while at the School of Computer Science, Carleton
University, Ottawa, Ontario, Canada},
    \email{sander@cg.scs.carleton.ca}}
}

\renewcommand{\figurename}{Figure}
\graphicspath{{figures/}}

\newcommand{\Rset}{\mathbb R}
\newcommand{\Sset}{\mathbb S}
\newcommand{\bd}{\partial}
\newcommand{\conv}{\mathop{\mathrm{conv}}}
\newcommand{\PARTITION}{\textsc{partition}\xspace}

\begin{document}
\maketitle

\begin{abstract}
Given a polygonal region containing a target point (which we assume is the origin), 
it is not hard to see that there are two points on the perimeter
that are antipodal, that is, whose midpoint is the origin. 
We prove three generalizations of this fact. 
(1)~For any polygon (or any compact planar set) containing the origin, 
it is possible to place a given set of weights on the boundary 
so that their barycenter (center of mass) coincides with the origin, 
provided that the largest weight does not exceed the sum of the other weights. 
(2)~On the boundary of any $3$-dimensional compact set containing the origin, 
there exist three points that form an equilateral triangle centered at the origin. 
(3)~For any $d$-dimensional bounded convex polyhedron containing the origin, 
there exists a pair of antipodal points 
consisting of a point on a $\lfloor d/2 \rfloor$-face and a point on a~$\lceil d/2\rceil$-face.
\end{abstract}

\section{Introduction}
\label{sec_introduction}

We will discuss three generalizations of the following observation (in
this paper, a polygon or a polyhedron is always closed and bounded).

\addtocounter{theorem}{-1}
\begin{theorem}
\label{theorem: antipodal}
On the perimeter of any polygon 
containing the origin, 
there are two points that are \emph{antipodal}, that is,
whose midpoint is the origin. 
\end{theorem}

In other words, we have 
\begin{equation*}
2 P \subseteq \bd P \oplus \bd P 
\end{equation*}
for any polygon~$P$, 
where 
$\bd P$ denotes its boundary, 
$A \oplus B = \{\, x + y \mid x \in A, \ y \in B \,\}$ 
is the Minkowski sum of regions $A$ and $B$, 
and 
$\alpha A = \{\, \alpha x \mid x \in A \,\}$ is the copy of $A$ 
scaled (about the origin) by a real number~$\alpha$. 

\begin{proof}[Proof of Theorem~\ref{theorem: antipodal}]
Consider $-P$, the copy of the given polygon~$P$ reflected about the origin. 
Since $P$ and $-P$ cannot be properly contained in the other 
(and they both contain the origin), 
their boundaries intersect 
at some point $q \in \bd P \cap (- \bd P)$. 
Then $q$ and $-q$ form the desired pair of points. 
\end{proof}

\subsection*{Distinct weights}

An interpretation of Theorem~\ref{theorem: antipodal} is that 
we can put two equal weights on the perimeter
and balance them about the origin. 
Generalizing this to different sets of weights, 
we prove the following 
in Section~\ref{sec_2D}
(note that this subsumes Theorem~\ref{theorem: antipodal}). 

\begin{theorem}
\label{theorem: different weights}
Suppose that $k$ weights $w _1 \geq w _2 \geq \dots \geq w _k$ satisfy 
$w _1 \leq w _2 + \dots + w _k$. 
Then for any polygon 
(or any compact set)
$P \subseteq \Rset ^2$
containing the origin, 
the weights can be placed on the boundary~$\bd P$
so that their center of mass is the origin. 
\end{theorem}

In terms of the Minkowski sum, the theorem says that 
\begin{equation*}
 (w _1 + \dots + w _k) P 
\subseteq 
 w_1 \bd P \oplus w_2 \bd P \oplus \dots \oplus w_k \bd P
\end{equation*}
if none of the weights is bigger than the sum of the rest. 

If $P$ is the unit disk, 
Theorem~\ref{theorem: different weights} is related to 
a reachability problem of a chain of links (or a robot arm)
of lengths $w_1$, $w_2$, \dots, $w_k$ 
where one end is placed at the origin, 
each link can be rotated around the joints, 
and the links are allowed to cross each other. 
In order to reach every 
point of the disk of radius $w _1 + \dots + w _k$ centered at the origin, 
it is known that the 
condition 
$w_1 \le w_2 + \dots + w_k$
is sufficient (and necessary)~\cite{hopcroft1985movement}. 
Theorem~\ref{theorem: different weights} generalizes this to 
arbitrary $P$.

Our proof is constructive and leads to an efficient algorithm to find 
such a location of points for a given polygon~$P$.
On the other hand, if we drop 
the condition $w_1 \le w_2 + \dots + w_k$, 
then the conclusion does not hold in general
(just let $P$ be a disk centered at the origin), 
and we show that 
it is NP-hard to decide 
whether it holds for a given polygon~$P$.

\subsection*{Tripodal points}

Our next result
concerns the $3$-dimensional setting.
Generalizing the notion of antipodal points
in Theorem~\ref{theorem: antipodal}, 
we prove the following in Section~\ref{sec_3d}. 

\begin{theorem}
\label{theorem: tripodal}
On the boundary of any $3$-dimensional compact set
containing the origin, 
there are \emph{tripodal} points, 
that is, three points forming an \emph{equilateral} triangle 
centered at the origin. 
\end{theorem}

After publication of the conference version of the present
paper~\cite{barba2014weight}, we discovered that Theorem~\ref{theorem:
  tripodal} follows from a result by Yamabe and
Yujob\^o~\cite{yamabe1950continuous}, with an extension to star-shaped polyhedra
in dimensions higher than three studied by Gordon~\cite{gordon1955generalization}.  Since
the first of these papers is hard to find and the second is not
available in English, we present our alternative proof.

A classical problem reminiscent of Theorem~\ref{theorem: tripodal} 
is the square peg problem of Toeplitz.
Given a closed curve in a plane, the problem asks for four points on
the curve forming the vertices of a square.
It was conjectured by Otto Toeplitz in 1911 
that every Jordan curve contains four such points.
Although the problem is 
still open for general Jordan curves, 
it has been affirmatively solved for curves with some smoothness conditions.
As a variant of this problem, 
Meyerson~\cite{meyerson1980equilateral} and
Kronheimer and Kronheimer~\cite{kronheimer1981tripos} 
proved that
for any triangle $T$ and any Jordan curve $C$, 
we can find three points on $C$ forming the vertices of a 
triangle similar to $T$
(note the contrast to our Theorem~\ref{theorem: tripodal} where 
we need the triangle to be equilateral). 
See a recent survey of Matschke~\cite{matschke2014survey} on these problems.

\subsection*{Antipodal points on convex polyhedra}

By viewing Theorem~\ref{theorem: antipodal} again as the
balancing of two equal weights, 
we can consider another generalization to convex polyhedra in dimension~$d$,
asking whether there are two antipodal points on the surface of the polyhedron.
This is not very interesting 
if we are allowed to put them anywhere on the surface of the polyhedron: 
we can then cut the polyhedron by any plane through the origin
and apply Theorem~\ref{theorem: antipodal}.

The question becomes interesting 
if we restrict the points to lie on \emph{lower-dimensional faces} of the polyhedron. 
In Section~\ref{section: convex} we prove the following.
\begin{theorem}
\label{theorem: halving}
For any convex polyhedron $P \subset \Rset^d$ containing the origin,
there is an antipodal pair consisting of a point on a $\lfloor d/2
\rfloor$-face and a point on a~$\lceil d/2\rceil$-face.
\end{theorem}

In other words, 
\begin{equation*}
2P \subseteq S_{\lfloor d/2 \rfloor}(P) \oplus S_{\lceil d/2 \rceil}(P). 
\end{equation*}
where $S _k (P)$ denotes the $k$-skeleton of a convex polyhedron~$P$.

We also show that it is not possible to replace the pair of dimensions
by~$(k, d-k)$ for any~$k < \lfloor d/2\rfloor$.

By repeated application of Theorem~\ref{theorem: halving}, it follows
that when the dimension~$d$ is a power of two, then there are~$d$
points on the \emph{edges} (the one-skeleton) of~$P$ whose barycenter
is the origin.  Dobbins has shown using equivariant topology that this
statement does indeed hold in any dimension, and in a more general
form: When $d = nk$, then there are $n$ points on the~$k$-skeleton
whose barycenter is the origin~\cite{dobbins2015point}.  Blagojevi\'c et
al.~\cite{blagojevic2019barycenters} gave an alternative proof of the same
result.  Finally, Dobbins and Frick~\cite{dobbins2021barycenters}
generalize Theorem~\ref{theorem: halving} by setting~$d = nk + r$, for
$0 \leq r < n$, and prove the existence of~$n$ points in the~$k$-faces
and~$(k+1)$-faces whose barycenter is the origin.  They also consider
non-equal weights.

\subsection*{Related work}

Bringing the center of mass to a desired point 
by putting counterweights is a common technique 
for reduction of vibrations in mechanical engineering~\cite{arakelian2005shaking}. 
There have been studies on Minkowski operations 
considering the boundary of objects 
(see, for instance, Ghosh and Haralick~\cite{ghosh1996mathematical}), 
but our paper seems to be the first 
to deal with the general question of covering the body 
with convex linear combinations of the boundary.

\section{Distinct weights}
\label{sec_2D} 

We prove Theorem~\ref{theorem: different weights}.  Let us first
assume that $P$ is a simple polygon containing the origin in its
interior, and let $p$ be a point on~$\bd P$ closest to the origin. We
first put the biggest weight $w _1$ at~$p$, and the remaining~$k-1$
weights at $p' = -p \cdot w _1 / (w _2 + \ldots + w _k)$.  By the
assumption $w _1 \leq w _2 + \dots + w _k$ we have $p'\in P$.  The
barycenter of the chosen~$k$ weighted points is the origin.  One of
the weights, namely~$w_1$, lies on~$\bd P$, while the remaining
weights lie in~$P$.

We will now repeatedly move two of the weights while maintaining the
barycenter at the origin, in each step moving one more weight to~$\bd
P$.  Let $q_1, q_2, \dots, q_k$ be the current position of the~$k$
weights, with $q_1, q_2, \dots, q_i \in \bd P$, while $q_{i+1}= \dots =
q_{k} = p'$.

We will now move $q_i$ and~$q_{i+1}$ such that both lie on~$\bd P$.
We set $r = \sum_{j=1}^{i-1}w_{j}q_{j} + \sum_{j=i+2}^{k}
w_{j}q_{j}$. By assumption, we have $w_{i}q_{i} + w_{i+1}q_{i+1} + r =
0$, and thus $q_{i+1} = -(r + w_{i}q_{i})/w_{i+1}$.  If we move the
weight~$w_{i}$ to a new position~$q'_{i}\in \bd P$, we can move
$w_{i+1}$ to $q'_{i+1} = -(r + w_{i}q'_{i})/w_{i+1}$ to maintain the
barycenter at the origin.  As $q'_{i}$ moves along~$\bd P$, the
point~$q'_{i+1}$ moves along~$\bd P'$, where $P' = -(w_{i}/w_{i+1}) P
- r/w_{i+1}$, that is, a translated, reflected, and scaled copy
of~$P$.  Since~$q_{i+1} = p' \in \bd P' \cap P$, the boundary~$\bd P'$
contains a point in~$P$. Since $w_{i}/w_{i+1} \geq 1$, $P'$ cannot lie
entirely inside~$P$, and so the boundaries~$\bd P$ and~$\bd P'$ must
intersect in a point~$q'_{i+1}$.  We move $w_{i+1}$ to~$q'_{i+1} \in
\bd P$, and move~$w_{i}$ to the corresponding point~$q'_{i}\in\bd P$,
see \figurename~\ref{fig:twoweights}.
\begin{figure}[t]
  \centerline{\includegraphics{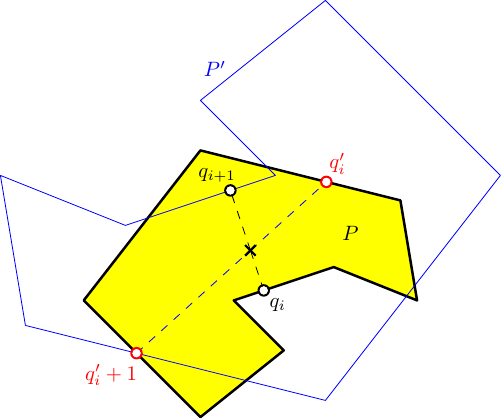}}
  \caption{Proof of Theorem~\ref{theorem: different weights}.  The
    weights $w_i$ and $w_{i+1}$ are initially at $q_i \in \bd P$ and
    $q_{i+1} = p' \in P$. As $w_i$ moves along $\bd P$, $w_{i+1}$
    moves along a magnified (and reflected) copy~$\bd P'$ of $\bd P$,
    which intersects $\bd P$ at some point~$q'_{i+1}$.}
\label{fig:twoweights}
\end{figure}

Repeating this step $k-1$~times, we bring all weights to $\bd P$,
proving Theorem~\ref{theorem: different weights} for the case
where~$P$ is a simple polygon.

We now consider the case that~$P$ is an arbitrary compact set
containing the origin in~$\Rset^{2}$.  For $m = 1, 2, \dots$ we
partition the plane with an axis-aligned grid whose cells have side
length~$1/m$, and let $A _m \supseteq P$ be the union of all grid
cells intersecting~$P$ (where a ``grid cell'' is to be understood as
including its boundary). Note that each point in $\bd A _m$ is within
distance $\sqrt 2 / m$ from $\bd P$.  Let~$B _m \subseteq A _m$ be the
union of all grid cells reachable from the origin by a path in the
interior of $A _m$ (this step is necessary because we do not
require~$P$ to be connected). Let $X_{m}$ be the unique unbounded
connected component of~$\Rset^2 \setminus B_{m}$, and set~$C_{m} =
\Rset^2 \setminus X_{m}$ (in other words, $C_{m}$ is~$B_{m}$ with all
``holes filled in'').

Since $\bd C_{m} \subseteq \bd B _m \subseteq \bd A_{m}$, each point
in~$\bd C_{m}$ lies at distance at most~$\sqrt{2}/m$ from~$\bd P$. We
observe that~$C_{m}$ is a simple polygon containing the origin, and so
we can apply the above special case and obtain a $k$-tuple of
locations $\mathbf q ^{(m)} = \bigl( q ^{(m)} _1, \ldots, q ^{(m)} _k
\bigr) \in (\bd C_m)^k$ such that putting the weight $w _i$ at 
$q^{(m)}_i$ (for $i = 1$,\ldots, $k$) brings the barycenter to the origin.

Since the sequence $\mathbf q ^{(1)}$, $\mathbf q ^{(2)}$, \ldots\ is
in the compact space $U ^k$, where $U$ is a sufficiently large compact
set containing $P$, it has a subsequence that converges to some
$\mathbf q = (q _1, \ldots, q _k)$. Since each $q ^{(m)} _i$ is within
distance $\sqrt 2 / m$ from $\bd P$, each $q _i$ is in $\bd
P$. Furthermore, since the barycenter is a continuous function of the
location of the weights, we conclude that putting the weight $w _i$ at
$q _i$ for each $i$ brings the barycenter to the origin, proving the
general form of Theorem~\ref{theorem: different weights}.

\subsection*{Algorithmic aspects}

We consider the computational problem
that corresponds to Theorem~\ref{theorem: different weights}: 
Given a region $P$ and a set of $k$ weights, 
we want to determine whether we can 
balance the weights by putting them on $\bd P$, 
and if so, to find such a location. 
We restrict ourselves to the case where $P$ is a simple polygon with $n$
vertices, and design algorithms in terms of $n$ and $k$.

If none of the weights exceeds the sum of the others, 
the proof of Theorem~\ref{theorem: different weights} 
implies a polynomial-time algorithm.
In order to replace $q_i$ and $q_{i+1}$ 
with a pair of boundary points $q'_i$ and $q'_{i+1}$
(Figure~\ref{fig:twoweights}), 
we need to find an intersection point of 
$\bd P$ and~$\bd P'$.
This can be done in $O(n \log n)$ time.
The initial location of the largest weight 
can be found in $O(n)$ time. Thus, we have an $O(k n \log n)$ time algorithm.

We can design a faster algorithm as follows. 
We greedily divide the weights into three groups so that 
no group weighs more than the sum of the rest.
This is always possible in $O(k)$ time
(as long as no single weight exceeds the sum of the others). 
Thus, we have an instance for $k = 3$, 
which we solve in $O (n \log n)$ time. 
This gives an $O(k + n \log n)$-time algorithm,
although the output may look a little artificial 
since all weights will be located at (at most) three points. 

On the other hand, if we are given a set of an unknown number of weights 
that may contain a weight exceeding the sum of the rest, 
then the problem is NP-hard.
\begin{proposition}
\label{prop:npcomplete}
There exists a polygon $P \subseteq \Rset^2$ containing the origin 
such that it is NP-hard to determine if 
a given set of weights can be placed on the boundary $\bd P$
so that their barycenter is at the origin. 
\end{proposition}

\begin{proof}
We prove NP-hardness by reducing the \PARTITION problem to this problem.
The input to \PARTITION is a set of $N$ nonnegative integers $a_1, a_2, \ldots a_N$,
and the problem asks whether there is a subset $X \subset \{1,\dots,N\}$ such that
\[\sum_{i \in X} a_i = \sum_{i \in \{1,\dots,N\} \setminus X} a_i.\]

We transform the problem into a weight balancing problem as follows:
we set $k=N+1$, $w_i = a_{i-1}$ for $i=2,3,\dots,k$, and $w_1 = 2
\sum_{i=1}^{N} a_i$.

Let $P$ be the non-convex polygon with vertices 
$(0,1)$, $(2,2)$, $(2,-2)$, $(0,-1)$, $(-2,-2)$, and $(-2,2)$,
see \figurename~\ref{fig:npcomplete}(left). 
Note that $P = -P$ and $-2P$ contains the convex hull $\conv(P)$ of $P$. 
Moreover, the two reflex vertices of $-2P$ are the only points of $\bd
2P\cap \conv(P)$, and 
each of the points is the midpoint of an edge of $\conv(P)$ as shown
in
\figurename~\ref{fig:npcomplete}(right). 
\begin{figure}[t]
  \centerline{\includegraphics{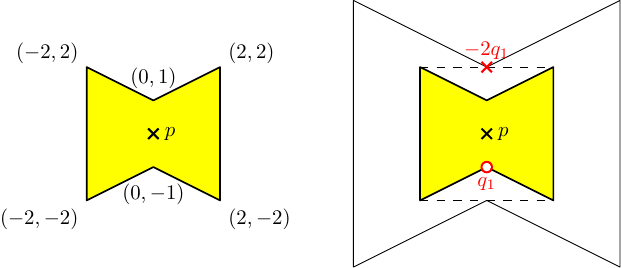}}
  \caption{A hard instance of the balancing location problem.}
  \label{fig:npcomplete}
\end{figure}

Observe that the only possible location $q_1$ of weight $w_1$ is one
of the two reflex vertices since its reflection $-2q_1$ can be written
as a convex combination of other points on $\bd P$, and is hence
contained in $\conv(P)$.  That is, $-2q_1$ lies on the midpoint of an
edge of $\conv(P)$ (without loss of generality, we may assume it is
the edge $e$ from $(-2,2)$ to $(2,2)$). In particular, there is a
solution if and only if we can place the remaining points in a way
that their barycenter lies on the midpoint of~$e$.

Since the new target point lies on the edge $e$ of $\conv(P)$, the
only possible location for the remaining points is $(-2,2)$ or
$(2,2)$. Moreover, the barycenter becomes the midpoint if and only if
the weights are equally divided. Thus, the \PARTITION problem is
reduced to the balancing location problem. Since \PARTITION is
NP-complete, detecting the existence of a balancing location is
NP-hard.
\end{proof}

\section{Tripodal points}
\label{sec_3d}

In this section, we 
consider a $3$-dimensional compact set~$P$
and prove Theorem~\ref{theorem: tripodal}, 
which states that there are tripodal points on the boundary $\bd P$.
Note that 
tripodal points are a natural analogue of antipodal points: 
saying that three points are tripodal is equivalent to 
requiring that 
they are at the same distance from the origin and 
their barycenter is the origin. 

We will first assume that~$P$ is the union of a finite number of
convex polyhedra such that the boundary~$\bd P$ is connected.
Let $p_0$ and $p_1$ be a nearest and a farthest point on $\bd P$, respectively, from the origin~$o$.
They exist because $\bd P$ is compact.
By our assumption on~$P$, there is a simple piecewise-linear path $L$ from $p_0$ to $p_1$ on $\bd P$, 
parametrized by a one-to-one continuous function 
$\gamma \colon [0,1] \to L$ 
such that $\gamma(0) = p_0$ and $\gamma(1) = p_1$.

We claim that 
there exist three points 
$a \in L$, $b \in \bd P$, and $c \in \bd P$
that are tripodal.
For each $q \in L$, 
let $H (q)$ be the set of vectors perpendicular to 
the line through $o$ to $q$. 
We use the following fact.

\begin{lemma}
\label{claim:pach}
There exists a continuous piecewise algebraic 
function $\mathbf v\colon L \to \Sset ^2$ such that
$\mathbf v (q) \in H (q)$ for all $q \in L$.
\end{lemma}
\begin{proof}
Consider the points of $L$ as vectors in~$\Rset ^3$, and let $\mathcal
F \subset \Sset^2$ be the set obtained by normalizing all those
vectors.  Since $L$ consists of a finite number of linear segments,
$\mathcal F$ has measure zero.  Thus, there must exist some vector
$\mathbf w _0 \in \Sset ^2 \setminus (\mathcal F \cup -\mathcal F)$.

We define the function $\mathbf v$ as follows: for any point~$q \in
L$, let $\mathbf v (q)$ be the normalized projection of~$\mathbf w _0$
into $H (q)$.  By construction, $\mathbf w _0$ is not parallel to any
vector of~$L$, hence its projection is nonzero and its normalization
is properly defined.  Since~$H(q)$ is a continuous function of~$q$,
the projection of~$\mathbf w _0$ on~$H(q)$ is continuous. Furthermore,
on a single linear segment of~$L$ this projection is an algebraic
function of~$q$.  Since normalization is continuous and algebraic, the
function~$\mathbf v$ is continuous and piecewise algebraic.
\end{proof}

We fix such a function $\mathbf v\colon L\to \Sset ^2$.
For each $t \in [0, 1]$ and each angle $\theta \in [0, \pi]$, 
let $b (t, \theta)$ and $c (t, \theta)$ be the unique pair of points
such that 
$\gamma (t)$, $b (t, \theta)$, $c (t, \theta)$ are tripodal points and
the vector $b (t, \theta) - c (t, \theta) \in H (\gamma (t))$ 
makes an angle of $+\theta$ with $\mathbf v(\gamma (t))$ (using the
vector~$\gamma(t)$ to determine the sign of the rotation).
Define $f _1 (t, \theta) \in \{+, -, 0\}$ by whether the point 
$b (t, \theta)$ lies inside (the interior of) $P$, outside $P$, 
or on $\bd P$.
Define $f _2 (t, \theta)$ analogously using the point $c (t, \theta)$.

If there is $(t, \theta)$ such that 
$f _1 (t, \theta) = f _2 (t, \theta) =0$, 
then $\gamma (t)$, $b (t, \theta)$, $c (t, \theta)$ are tripodal points
and we are done. 
Suppose otherwise. 
Define the \emph{signature} 
of $(t, \theta)$, denoted 
$F (t, \theta)$, as
\begin{equation*}
F (t, \theta) = \begin{cases}
++ & \text{if} \ (f _1 (t, \theta), f _2 (t, \theta)) \in \{(+, +), (+, 0), (0, +)\}, \\
-- & \text{if} \ (f _1 (t, \theta), f _2 (t, \theta)) \in \{(-, -), (-, 0), (0, -)\}, \\
+- & \text{if} \ (f _1 (t, \theta), f _2 (t, \theta)) = (+, -), \\
-+ & \text{if} \ (f _1 (t, \theta), f _2 (t, \theta)) = (-, +)
\end{cases}
\end{equation*}
(\figurename~\ref{fig:signature}). 
\begin{figure}
\centering
\includegraphics{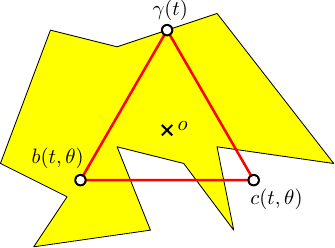}
\caption{The signature is $+-$ for this $(t, \theta)$.}
\label{fig:signature}
\end{figure}
Since $p_0$ and $p_1$ are the nearest and the farthest points, 
it holds that
$F(0, \theta) = ++$ and
$F(1, \theta) = --$.

Consider the domain~$[0,1] \times [0, \pi]$ of~$F(t, \theta)$. It is
partitioned into regions corresponding to the four different values
of~$F$. Boundaries of the regions consist of points where either~$b(t,
\theta)$ or~$c(t, \theta)$ lies on~$\bd P$.  By our assumption,
regions for~$++$ and~$--$ cannot have a common boundary point, and
neither can regions for~$+-$ and~$-+$.  Since $\mathbf v$ is
piecewise algebraic, we can choose a parameterization where the
region boundaries are piecewise algebraic, so any line through the
domain intersects a finite number of points, and regions have a finite
number of points with a tangent parallel to the~$t$-axis.

For each $\theta$, consider the transition of $F (t, \theta)$ as $t$
changes from $0$ to $1$.  This corresponds to the finite sequence of
regions intersected by the segment~$\{ (t, \theta) \mid t \in [0,
  1]\}$, and we obtain a finite walk $\mathcal W (\theta)$ from $++$
to $--$ in the graph~$C$ shown in \figurename\ \ref{fig:graph}.

\begin{figure}
\centering
\includegraphics{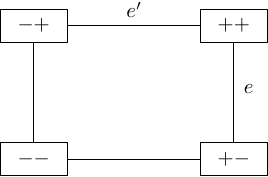}
\caption{The cycle $C$.}
\label{fig:graph}
\end{figure}

Consider the edge $e$ between $++$ and $+-$. 
A walk from $++$ to $--$ is called \emph{even} 
if it uses $e$ an even number (possibly zero) of times, 
and it is called \emph{odd} otherwise. 
For example, the path $++, +-, --$ is odd and 
$++, -+, --$ is even. 

As we increase $\theta$ continuously from~$0$ to~$\pi$, the walk
$\mathcal{W} (\theta)$ changes when the sequence of regions
intersected by the segment~$\{ (t, \theta) \mid t \in [0, 1]\}$
changes, that is, for a finite number of values of~$\theta$ where the
segment is tangent to one of the regions.  The walk can change in two
possible ways:
\begin{itemize}
\item when the segment starts to intersect a new region, an entry
  $a$ in $\mathcal{W}(\theta)$ is replaced by a sequence $a$, $b$,
  $a$, where $b$ is a neighbor of $a$ in $C$;
\item when the segment stops intersecting a region, a sequence
  $a$, $b$, $a$ is replaced by $a$.
\end{itemize}
Neither of these events change the parity of the walk.

On the other hand, 
the walks $\mathcal{W} (0)$ and $\mathcal W (\pi)$ have different parities.
To see this, let $e'$ be the edge between $++$ and $-+$. 
Since $e$ and $e'$ form a cut 
separating $++$ and $--$ in $C$, 
each walk must use them an odd number of times in total. 
Since $b (t, 0) = c (t, \pi)$ and $c (t, 0) = b (t, \pi)$, 
the walk $\mathcal W (\pi)$ is obtained from $\mathcal W (0)$
by exchanging $-+$ and $+-$, 
and hence has opposite parity. 

This is a contradiction, so Theorem~\ref{theorem: tripodal} follows
for the special case where~$P$ is the union of a finite number of
convex polyhedra with connected boundary.

Consider now the general case, where~$P$ is an arbitrary compact set
containing the origin in~$\Rset^{3}$.  We proceed as in the
two-dimensional result: for $m = 1, 2, \dots$ we partition~$\Rset^3$
with an axis-aligned grid of width $1/m$, let $A _m \supseteq P$ be
the union of all grid cells intersecting~$P$, let~$B _m \subseteq A
_m$ be the union of all grid cells reachable from the origin by a path
in the interior of $A _m$, let $X_{m}$ be the unbounded connected
component of~$\Rset^3 \setminus B_{m}$, and set~$C_{m} = \Rset^3
\setminus X_{m}$.

Unlike in the two-dimensional case, the set~$C_{m}$ is not necessarily
a topological ball.  However, we observe that $C_{m}$ is the union of
a finite number of convex polyhedra (the grid cells), and since $\bd
C_{m} = \bd X_{m}$, its boundary~$\bd C_{m}$ is connected.  This
implies that we can apply the above special case and obtain tripodal
points~$(a_{m}, b_{m}, c_{m}) \in (\bd C_{m})^{3}$.  By compactness,
the sequence $(a_{m}, b_{m}, c_{m})$ has a subsequence converging to a
triple of points~$(a, b, c)$.  By continuity, $a, b, c$ are tripodal
points.  Since~$\bd C_{m} \subseteq \bd B _m \subseteq \bd A_{m}$,
each point in~$\bd C_{m}$ lies at distance at most~$\sqrt{3}/m$
from~$\bd P$, so compactness of~$\bd P$ implies that~$a, b, c \in \bd
P$.  This completes the proof of the general form of
Theorem~\ref{theorem: tripodal}.

We may ask similar questions for three points forming other shapes, 
or for higher dimensions. 

\begin{conjecture}\label{conj:simplex}
For a $d$-dimensional polyhedron $P$ containing the origin,
there exist $d$ points on $\bd P$ forming a 
regular $(d-1)$-dimensional simplex centered at the origin. 
\end{conjecture}

Algorithmic aspects need further investigation.
It is easy to devise an $O (n^3)$-time algorithm to find a tripodal location
guaranteed by Theorem~\ref{theorem: tripodal} 
for a polyhedron with $n$ vertices, 
just by going through all the triples of faces. 
It is not clear if this can be improved.

\section{Antipodal points on convex polyhedra}
\label{section: convex}

A $d$-dimensional (closed bounded) convex polyhedron $P$ decomposes
into faces of dimensions $i = 0,1,\ldots, d$.  Let $F_i$ be the set of
$i$-dimensional faces.  The union $S_k (P) = \bigcup_{i=0}^{k}
\bigcup_{f\in F_i} f$ of faces of at most $k$ dimensions is called the
\emph{$k$-skeleton} of $P$.  In particular, the $1$-skeleton $S_1(P)$
is the union of edges (including vertices), and the $(d-1)$-skeleton
is $\bd P$.

We will now prove Theorem~\ref{theorem: halving} by proving
\begin{equation*}
2P \subseteq S_{\lfloor d/2 \rfloor}(P) \oplus S_{\lceil d/2 \rceil}(P). 
\end{equation*}

Choose any point of the left-hand side, $2 P$. 
We will show that this point is in the right-hand side. 
We may assume that this point is in the interior of $2 P$, 
since the right-hand side is a closed set. 
Also, without loss of generality, 
we may assume that this point is the origin. 
Thus, assuming that $P$ contains the origin in its interior, 
we need to show that the origin belongs to the right-hand side, 
or equivalently, that $
S_{\lfloor d/2 \rfloor}(P) \cap S_{\lceil d/2 \rceil}(-P)
$ is nonempty. 

For simplicity of notation, we assume that $d$ is even. 
The odd case is shown identically 
by replacing $d/2$ by 
$\lfloor d/2 \rfloor$ and $\lceil d/2 \rceil$ accordingly. 

Since $P$ contains the origin in its interior, 
the intersection $P \cap (-P)$ is a $d$-dimensional convex polyhedron. 
Moreover, its boundary $C$ is centrally symmetric (i.e., $C = -C$). 
It suffices to show that $C$ has a vertex in $S_{d/2}(P) \cap S_{d/2}(-P)$. 

A facet ($(d-1)$-dimensional face) of $C$ is a subset of a facet of either $P$ or $-P$.
We start with the special case in which $C$ is simple. That is, every vertex of $C$ is contained in 
exactly $d$ facets of $P$ or $-P$.   
A vertex of $C$ is of \emph{type $(j, d-j)$} if it is contained in $j$ facets of $P$ and $d-j$ facets of $-P$.
Let $v$ be any vertex of $C$, and let $(k, d-k)$ be the type of $v$.  If $k =d/2$, we are done. 
Thus, we assume without 
loss of generality that $k< d/2$.  Since $C$ is centrally symmetric, $-v \in C$, and $-v$ is of type $(d-k, k)$.
Since the $1$-skeleton of $C$ is connected, there exists a path $P$ in the skeleton from $v$ to $-v$.
Let $(x,y)$ be an edge of $P$ with $x$ and $y$ of type $(i, d-i)$ and $(j, d-j)$, respectively.
Then $j \in \{i-1, i, i+1\}$.  Thus, there exists a vertex $w$ on $P$ of type $(d/2, d/2)$.

Now, we consider the general case where $C$ might have a vertex that is an intersection of more than $d$ facets.
We consider an infinitesimal perturbation of hyperplanes defining facets of $P$ to make $C$ simple.
Then, the perturbed version $\tilde{C}$ of $C$ has a vertex $\tilde{v}$ of type $(d/2, d/2)$, which corresponds to a vertex $v$ of $C$.
Thus, $v$ must lie at an intersection of $S_{d/2}(P)$ and
$S_{d/2}(-P)$, completing the proof of Theorem~\ref{theorem: halving}.

Thus we can always find an antipodal pair of points from 
$\lfloor d/2 \rfloor$- and $\lceil d/2 \rceil$-dimensional faces.
However, this does not extend to 
other pairs of dimensions $k$ and $d-k$. 

\begin{proposition}
There exists a convex polyhedron $P \subseteq \Rset^d$ containing the origin
such that for any $k < \lfloor d/2 \rfloor$, 
it holds that $S_{k}(P) \cap S_{d-k}(-P) = \emptyset$.
\end{proposition}
\begin{proof}
First, we consider the case where $d=2m$ is even (thus, $k<m$).
Consider an equilateral triangle $T$ centered at the origin. 
Then, we observe that all three vertices of $T$ lie outside $-T$.
Let $T^{m} = T \times \dots \times T$ be the Cartesian product of $T$ in $\Rset ^{2m}$.
Then, a $k$-dimensional face of $P = T^{m}$ is the Cartesian product of $k$ edges and $m-k$ vertices of $T$.
Since $m-k > 0$ and a vertex of $T$ lies outside $-T$, the face cannot intersect $-P = (-T)^{m}$. 
If $d= 2m+1 \ge 3$ is odd, we consider $P = I \times T^{m}$, where $I = [-1, 2] $ is an interval. 
The remaining argument is analogous.
\end{proof} 

As mentioned in the introduction, repeated application of
Theorem~\ref{theorem: halving} immediately gives us the following:
\begin{proposition}
\label{prop:2power}
Let $k$ be a positive integer and let $d \le 2^k$.
Then, on the $1$-skeleton of any $d$-dimensional convex polyhedron, 
there are $2^k$ points whose barycenter is at the origin. 
\end{proposition}

\begin{proof}
We use induction on $k$. 
The statement is true for $k = 1$ (Theorem~\ref{theorem: antipodal}). 
It follows from Theorem~\ref{theorem: halving} 
that there are antipodal points $x \in F$ and $-x \in F'$, 
where $F$ and $F'$ are faces  
from $S_{\lfloor d/2 \rfloor}(P)$ 
and $S_{\lceil d/2 \rceil}(P)$, respectively. 
By the induction hypothesis applied to $F$ and $F'$ (translated by $-x$ and $x$), 
we have $2 ^{k - 1}$ points on the skeleton of $F$ with barycenter~$x$, 
and $2 ^{k - 1}$ points on the skeleton of $F'$ with barycenter~$-x$. 
These $2 ^k$ points together satisfy our requirement.
\end{proof}

We note that our method is constructive, and such a location of points
can be computed in polynomial time for any fixed dimension.  

Algorithmic aspects of the generalization of
Proposition~\ref{prop:2power} by Dobbins~\cite{dobbins2015point} appear to
be unexplored.

\section*{Acknowledgements}
This work was initiated at the Fields Workshop
on Discrete and Computational Geometry 
(15th Korean Workshop on Computational Geometry),
Ottawa, Canada, in August 2012,
and it was continued at 
IPDF 2012 Workshop on Algorithmics on Massive Data, held in
Sunshine Coast, Australia, August 23--27, 2012, 
supported by the University of Sydney International Program Development Funding,
a workshop in Ikaho, Japan,
and 
16th Korean Workshop on Computation Geometry in Kanazawa, Japan.
The authors are grateful to the organizers of these workshops.
They also thank John Iacono and Pat Morin for 
pointing out Proposition~\ref{prop:npcomplete},
as well as Jean-Lou~De~Carufel and Tianhao~Wang
for helpful discussions.

This work was supported by
NSF China (no.~60973026);
the Shanghai Leading Academic Discipline Project (no.~B114);
the Shanghai Committee of Science and Technology (nos.~08DZ2271800 and 09DZ2272800);
the Research Council (TRC) of the Sultanate of Oman;
NRF grant 2011-0030044 (SRC-GAIA) funded by the government of Korea;
the National Research, Development and Innovation Office, NKFIH, project KKP-133864;
the Austrian Science Fund (FWF), grant~Z~342-N31;
ERC Advanced Grant "GeoScape" (no.~882971);
the Ministry of Education and Science of the Russian Federation in the framework of MegaGrant (no.~075-15-2019-1926);
and 
the Natural Sciences and Engineering Research Council of Canada (NSERC).

\bibliographystyle{plain}
\bibliography{masscenter}
\end{document}